\documentclass{article}
\usepackage{amsthm, amsmath, amsfonts, amssymb}
\usepackage{cite}
\usepackage{xcolor}
\usepackage{caption}
\usepackage{subcaption}
\usepackage{graphicx}
\usepackage{mathtools}
\usepackage{algorithm}
\usepackage{algorithmic}
\usepackage{authblk}
\graphicspath{{figures/}}

\newtheorem{lemma}{Lemma}

\newcommand{\matr}[1]{\mathbf{#1}}

\title{Crossing Reduction of Sankey Diagram with Barycentre Ordering via Markov Chain}
\author[1,+]{He Chen LI}
\author[1,+]{Shi Ying LI}
\author[1]{Bo Wen TAN}
\author[1,2,*]{Shuai Cheng LI}
\affil[1]{Department of Computer Science, City University of Hong Kong, Kowloon Tong, Hong Kong }
\affil[2]{Department of Biology Engineering, City University of Hong Kong, Kowloon Tong, Hong
Kong}
\affil[+]{These authors contributed equally to this work.}
\affil[*]{To whom correspondence should be addressed.}

\date{\today}

\begin{document}
\maketitle

\begin{abstract}
Sankey diagram is popular for analyzing primary flows in network data. However, the growing complexity of data and hence crossings in the diagram begin to reduce its readability. In this work, we studied the NP-hard weighted crossing reduction problem of the Sankey diagram with both the common parallel form and the circular form. We expect to obtain an ordering of entities that reduces weighted crossings of links. We proposed a two-staged heuristic method based on the idea of barycentre ordering and used Markov chain to formulate the recursive process of obtaining such ordering. In the experiments, our method achieved 300.89 weighted crossings, compared with the optimum 278.68 from an integer linear programming method. Also, we obtained much less weighted crossings (87.855) than the state-of-art heuristic method (146.77). We also conducted a robust test which provided evidence that our method performed consistently against the change of complexity in the dataset.
\end{abstract}

\section{Introduction}
Sankey diagram depicts flows among entities in a system, where the thickness of a flow represents the flow quantity. Although Sankey diagram is originally introduced in \cite{sankey1898introductory} to display energy flows of a steam engine, its excellence in emphasizing dominant flows in a network makes it popular for analyzing sequence data, such as a vendor-to-customer network or page-viewing paths in Google Analytics (\cite{googleuse}). To better display data with increasing complexity, methods that automatically adjust the layout to enhance readability is in much desire. According to  \cite{Alemasoom2016}, a proper layout of the Sankey diagram should meet the following three criteria: minimum edge intersection, short-as-possible edge lengths and straight edges. In this paper, we focus on the first criteria: reducing the number of edge crossing.
    
The most common form of a Sankey diagram, called the \emph{parallel form}, is a horizontally-layered diagram with entities assigned to different vertical layers and flows with various thicknesses connecting entities between layers (normally from left to right). By viewing entities as vertices and flows edges, we can transform the hierarchical structure of a Sankey diagram into a layered graph. In this way, layout optimization of a Sankey diagram can be generalized as the layered graph drawing problem formulated by \cite{warfield1977crossing}, which is further decomposed into several sub-problems in \cite{sugiyama}. Among the sub-problems our work focuses on the NP-hard (proven in ~\cite{optimal}) crossing minimization problem where an ordering of vertices for each layer is determined to achieve minimum edge crossing. However, our problem differs from the classic crossing minimization problem in that each edge crossing is weighted as the intersected edges bear different weights from the thicknesses of flows.
    
For the classic crossing minimization problem (without considering edge weight), the most famous heuristic method is the barycentric (BC) method proposed by ~\cite{sugiyama}. In this method, each vertex is sorted  among its layer in ascending order of the barycentre of vertices connecting to it. For crossing reduction with weighted edges, \cite{Alemasoom2016} gives an approach combining the BC method and linear programming. This combined method uses the ordering produced by the BC method as input of the linear programming method whose objective function is minimizing the weighted sum of distances between connected nodes. That is, the linear programming method improves the placement of nodes within each layer determined previously by the BC method. However, the resultant layout is non-optimal since the heuristic BC method is used without considering edge weights. For exact solution, there is an integer linear programming (ILP) model with the objective being minimizing the sum of weighted crossing~\cite{optimal}. The attained optimal layout, in comparison to the layout obtained from the BC method, shows remarkable improvement on crossing reduction and readability improvement.

Another form of Sankey diagram with increasing popularity is the \emph{cycle form} with flows travelling in opposite direction (right to left) representing reversed data flows such as the recycle of resources. In this case, we study the crossing reduction problem for a specific type of the cycle form where the reverse flows only exist between the last (rightmost) layer and the first (leftmost) layer. There is not yet an algorithm to the crossing reduction problem on this type of Sankey diagram.

In this work, we first propose a two-staged heuristic method to reduce weighted crossing in the parallel form of the Sankey diagram (Section 2). In the first stage, we design a Markov Chain Method where we formulate the process of obtaining an ordering as a Markov chain and solve it with the eigenvector corresponds to the second largest eigenvalue of the Markov transition matrix. The solution is sufficient while non-optimal. For the second stage, we design a recursive Partition Refinement Method to further improve the ordering from the first stage. In this method, a vertex is given a range within the level and it gives different value within the range when used to calculate the barycentre of a connected vertex. We iterate the ordering in a back-and-force manner among layers until the positions of vertices remain unchanged.

In the following Section 3, we show a modified version of the above method that is applicable for reducing weighted crossing on our specified cycle form of the Sankey diagram. In the beginning, we ignore the connection between the first and the last level such that the formulated graph becomes parallel again, allowing us to obtain a partially calculated barycentre ordering. Subsequently, the second stage of this amended method undergoes a circular iteration route to also include the connection between the first and last levels.

In the Experiment Section, we first show the effect of our method on the parallel form by comparing the resultant ordering with those of the exact ILP method, the heuristic BC method and the combined method. The comparison includes both the visual effects as well as the weighted and non-weighted number of crossed edges. We find that our method is able to produce much better ordering than the two heuristic methods. In the process we also compare the difference between orderings produced in the two stages to demonstrate the effectiveness of both stages. For the cycle form, as there is no other methods for comparison, we apply our modified method on an artificial dataset with known optimal ordering. The result shows that in this case we are able to achieve optimal ordering even in the first stage. We also select a non-optimal output from stage 1 by using a different parameter set to verify the effectiveness of the second stage, which still produces the optimal ordering. Finally, we conduct a robustness test where we vary the complexity of the dataset and use the result from the ILP method as comparison. The test result verifies the stability of our method against the change of datasets. 

\section{Method on the Parallel Form}

We start by formulating a Sankey diagram in parallel form with $n$ layers as a $n$-level layered graph $G$. We regard each entity in the Sankey diagram as a vertex $v_i$ and let $V$ denote the set of all vertices in $G$. We say $v_i$ and $v_j$ belong to the same level in the graph if the corresponding entities lie in the same layer in the diagram. Set of vertices in the $i$-th level is denoted as $V_i$ and ${V_1, V_2, ..., V_n}$ form a partition of $V$. Without loss of generality, we assume that in a Sankey diagram all links are formed between successive layers. For links connecting nodes belonging to non-successive levels, we follow the practice in \cite{optimal} and \cite{Alemasoom2016} where dummy entities are added to all crossed levels such that the "long link" becomes the composition of several "short links" of the same thickness as the "long link" itself. Consequently, we say an undirected edge $(v_i, v_j)$ exists only if $v_i$ and $v_j$ belong to successive levels and are connected by a link in the diagram. We then have the edge set of $G$ as
\begin{equation*}
     E = \{(v_i, v_j) \mid v_i \in V_p, v_j \in V_{p+1}, p \in [1, n-1]\}.
\end{equation*}

With our assumption, $E$ can be partitioned into $n-1$ subsets where each subset $E_i$ is the set of edges connecting vertices between $V_i$ and $V_{i+1}$. Weight of an edge $(v_i, v_j)$, denoted as $w(v_i, v_j)$, follows from the thickness of the corresponding link in the Sankey diagram.

For the n-level layered graph $G = (V, E, n)$, its ordering $\sigma$ is the set $\{\sigma_1, ..., \sigma_n\}$ where $\sigma_i$ denotes the ordering of vertices in $V_i$.
With the formulated graph $G(V, E, n)$, our aim is to find an ordered graph $G(V, E, n, \sigma)$ with reduced weighted crossing. We measure the weighted crossing of an ordered graph by $\overline{K}$, the sum of production between weights of the crossed edges. Its calculation is a variation from the method of obtaining crossing number $K$ in \cite{warfield1977crossing} and is described in the following. Given ordering $\sigma$, we first define for each $E_i$ with a \emph{weighted interconnection matrix} $\matr{M^{(i)}}$ of size $|V_i| \times |V_{i+1}|$ where
\begin{equation}
    m^{(i)}_{j,k} = \left\{
\begin{array}{lcl}
w(v_j, u_k)     &      & if \; (v_j, u_k) \in E_i\\
0       &      & otherwise.
\end{array} \right.
\end{equation} 
In particular, we use $\matr{M^{(i)}_{j,:}}$ and $\matr{M^{(i)}_{:,k}}$ to denote the $j$-th row vector and transposed $k$-th column vector of $\matr{M^{(i)}}$.

To obtain the weighted crossing of $E_i$ for a particular pair of ordering $\sigma_{i}$ nad $\sigma_{i+1}$, we need to reorder the rows and columns $\matr{M^{(i)}}$ such that they comply with the given ordering. Therefore, we define for each ordering $\sigma_i$ a transformation matrix $\matr{A_i}$ with $a^{(i)}_{j,k}$ equals 1 if the $j$-th vertex in $\sigma_i$ has index $k$ and 0 otherwise. Then the transformed matrix, denoted as $\matr{\overline{M}^{(i)}}$, can be derived from the equation $\matr{\overline{M}^{(i)}} = \matr{A_{i}} \cdot \matr{M^{(i)}} \cdot \matr{A_{i+1}}$.

The weighted crossing of $E_i$ can therefore be calculated by the formula
\begin{equation}
    \overline{K}(\matr{\overline{M}^{(i)}}) = \sum^{|V_i|-1}_{j=1} \sum^{|V_i|}_{k=j+1} \big ( \sum^{|V_{i+1}|-1}_{p=1} \sum^{|V_{i+1}|}_{q=p+1} \overline{m}^{(i)}_{j,q} \times \overline{m}^{(i)}_{k,p} \big ).
\end{equation}
Subsequently, the total weighted crossing number of ordered graph $G$ with ordering $\sigma$ is defined as
\begin{equation}
    \overline{K}(G) = \sum^{n-1}_{i=1} \overline{K}(\matr{\overline{M}^{(i)}}).
\end{equation}

It has been shown in \cite{sugiyama} that \emph{barycentre ordering} can effectively reduce crossing number $K$. Here barycentre of a vertex is the weighted average of the position value of its connected vertices. A barycentre ordering places each vertex at its barycentre while avoiding the trivial case where all vertices share on barycentre value. To find such ordering, \cite{sugiyama} also gives the heuristic Barycentre(BC) Method in which weights among connected vertices are equal, regardless the different weights of the corresponding edges. However, in our case, crossing involving edge with larger weight will contribute more to the total weighted crossing number $\overline{K}(G)$ and therefore has higher priority to be avoided when deciding the ordering. That is, the weight of a connected vertex is proportional to the weight of the corresponding edge. Moreover, to define the positions of vertices, we view each level as a vertical line of height equal 1 and each vertex point on the vertical line takes a position value within $[0, 1]$. We further define the \emph{position vector} $\matr{u^{(i)}}$ of the $i$-th level where $u^{(i)}_j$ is the position of $v_j \in V_i$.

For a vertex $v_j \in V_i$ where $i \in [2, n-1]$, i.e. in one the the middle levels, all its connected vertices forms a neighboring vertex set $N(v_j)$. We further partition this set into the left neighboring vertex set $N_L(v_j)$ and the right neighboring set $N_R(v_j)$ containing vertices belonging to $(i-1)$-th level and the $(i+1)$-th level respectively. For $v_k$ in the first or last level, $N(v_j)$ consists of one one-sided neighboring set and an empty set for the other side. For each of the parted neighboring set, we have the vertex barycentre by the following equations
\begin{equation}
\label{eq:bc_l}
    B_{L}(v_j) = \left\{
\begin{array}{lcl}
\frac{1}{{||\matr{M^{(i-1)}_{:,j}}||}_1} \matr{M^{(i-1)}_{:,j}} \cdot u^{(i-1)}  &      & i \in [2, n]\\
B_{R}(v_j)       &      & i = 1,
\end{array} \right.
\end{equation}
\begin{equation}
\label{eq:bc_r}
    B_{R}(v_j) = \left\{
\begin{array}{lcl}
\frac{1}{{||\matr{M^{(i)}_{j,:}}||}_1} {\matr{M^{(i)}_{j,:}}}^{T} \cdot u^{(i+1)}  &      & i \in [1, n-1]\\
B_{L}(v_j)       &      & i = n.
\end{array} \right.
\end{equation}
where ${||\matr{x}||}_{1}$ is the $l_1$ norm of $\matr{x}$. Subsequently, we have the two-sided barycentre of $v_j$ as the average of the two one-sided barycentres:
\begin{equation}
    \label{eq:two_sided_barycentre}
    B(v_j) = \frac{B_{L}(v_j)+B_{R}(v_j)}{2}.
\end{equation}

Given position vector $\matr{u_{(i)}}$, $\sigma_i$ is the descending order of entries in $\matr{u_{(i)}}$ with the first vertex in the ordering placed uppermost in the level. Subsequently, to find a barycentre ordering is to find the corresponding position vectors where each entry is the barycentre of the corresponding connected vertices. To this end, we design a two-stage algorithm. In Stage 1, we introduce a Markov Chain Method which produces an ordering where most vertices satisfy the requirement for a barycentre ordering. In Stage 2, we give a Partition Refinement Method to refine the ordering from Stage 1 towards a complete barycentre ordering.

\subsection{Stage 1: the Markov Chain Method}

In Stage 1, we start with one-sided barycentre and eventually we want each vertex to have equal left and right barycentres so as to achieve two-sided barycentre. Given position vector $\matr{u^{(i)}}$ where $i \in [1, n-1]$, we can update the the position vector $\matr{u^{(i+1)}}$ such that $u^{(i+1)}_j = B_{L}(v_j)$ for each $v_j \in V_{i+1}$. Rewrite this process in matrix form, we have
\begin{equation}
    \label{eq:without_random}
    \matr{u^{(i+1)}} = \matr{L^{(i)}} \matr{u^{(i)}} 
\end{equation}
where $\matr{L^{(i)}}$ is the $i$-th \emph{left transition matrix} of size $|V_{i+1}|\times|V_i|$ transformed from $\matr{M^{(i)}}$
\begin{equation}
    \matr{L^{(i)}} = \matr{M^{(i)}} \cdot
    \begin{bmatrix} 
    \frac{1}{{{||\matr{M}_{:,1}||}}_1} & 0 & 0 & \dots & 0 \\
    0 & \frac{1}{{{||\matr{M}_{:,2}||}}_1} & 0 & \dots & 0 \\
    \vdots & \vdots & \vdots & \ddots & \vdots \\
    0 & 0 & 0 & \dots  & \frac{1}{{{||\matr{M}_{:,|V_{i+1}|}||}}_1}
    \end{bmatrix}
\end{equation}

However, we cannot determine an ordering if multiple entries in $\matr{u^{(i+1)}}$ have the same barycentre value. Consequently, we add a random matrix $\matr{S_{L^{(i)}}}$ of the same size as $\matr{L^{(i)}}$ with normalized rows (row entries summing to 1) to equation \ref{eq:without_random} by a factor $\alpha_1 \in [0, 1]$. Then equation \ref{eq:without_random} becomes
\begin{equation}
    \label{eq:with_random}
    \matr{u^{(i+1)}} =  \matr{\widetilde{L}^{(i)}}\matr{u^{(i)}}
\end{equation}
where we have the \emph{modified left transition matrix} as $\matr{\widetilde{L}^{(i)}} = (1-\alpha_1)\matr{L^{(i)}} + \alpha_1 \matr{S_{L^{(i)}}}$.

Similarly, given position vector $\matr{u^{(i)}}$ with $i \in [2, n]$, we update $\matr{u^{(i-1)}}$ such that $u^{(i-1)}_j = B_{R}(v_j)$ for $v_j \in V_{i-1}$. Here we define the \emph{modified right transition matrix} $\matr{\widetilde{R}^{(i)}} = (1-\alpha_1)\matr{R^{(i)}} + \alpha_1 \matr{S_{R^{(i)}}}$ where the original \emph{right transition matrix} $\matr{R^{(i)}}$ is again transformed from $\matr{M^{(i)}}$ via
\begin{equation}
    \matr{R^{(i)}} = 
    \begin{bmatrix} 
    \frac{1}{{{||\matr{M}_{1,:}||}}_1} & 0 & 0 & \dots & 0 \\
    0 & \frac{1}{{{||\matr{M}_{2,:}||}}_1} & 0 & \dots & 0 \\
    \vdots & \vdots & \vdots & \ddots & \vdots \\
    0 & 0 & 0 & \dots  & \frac{1}{{{||\matr{M}_{|V_{i}|,:}||}}_1}
    \end{bmatrix} \cdot \matr{M^{(i)}}. 
\end{equation}
Then we have the matrix form of the process as
\begin{equation}
    \label{eq:right}
    \matr{u^{(i-1)}} =  \matr{\widetilde{R}^{(i)}}\matr{u^{(i)}}.
\end{equation}

In addition, both modified transition matrices have the following properties: (1) all entries are non-negative as the entries of $\matr{M^{(i)}}$ are non-negative weights of edges, (2) all row sums equal 1. 

With the above, given position vector $\matr{u^{(1)}}$, with equation \ref{eq:without_random}, we have $\matr{u^{(2)}} = \matr{\widetilde{L}^{(1)}} \matr{u^{(1)}} $, which means that vertices in $V_2$ are placed based on the positions of nodes in $V_1$ such that each node in $V_2$ is at its left barycentre. Then, with $\matr{u^{(2)}}$ calculated, we can calculate $\matr{u^{(3)}}$ and so forth. Finally we can obtain
\begin{equation}
    \label{eq:l}
    \matr{u^{(n)}} = \matr{\widetilde{L}^{(n-1)}} \matr{u^{(n-1)}}
              = \matr{\widetilde{L}^{(n-1)}} \matr{\widetilde{L}^{(n-2)}} \matr{u^{(n-2)}}
              = \cdots
              = \Big(\prod_{i = 1}^{n - 1}{\matr{\widetilde{L}^{(n - i)}}}\Big)\matr{u^{(1)}}.
\end{equation} 
Therefore, by propagating the position vector of $\matr{u^{(1)}}$ thorough the product of the position matrix from $\matr{\widetilde{L}^{(n-1)}}$ to $\matr{\widetilde{L}^{(1)}}$, we are able to place all vertices in all levels in their left barycentres except for the first level based on $\matr{u^{(1)}}$.

On the other hand, given position vector $\matr{u^{(n)}}$, we can also get from equation \ref{eq:right} the following equation:
\begin{equation}
    \label{eq:r}
    \matr{u^{(1)}} = \matr{\widetilde{R}^{(1)}} \matr{u^{(2)}}
              = \matr{\widetilde{R}^{(1)}} \matr{\widetilde{R}^{(2)}} \matr{u^{(2)}}
              = \cdots
              = \Big(\prod_{i = 1}^{n - 1}{\matr{\widetilde{R}^{(i)}}}\Big)\matr{u^{(n)}}.  
\end{equation}
Note that in this way all vertices in $V_i$ with $i \in [1, n-1]$ are placed in their right barycentres based on $\matr{u^{(n)}}$.

From above, we can formally describe the Markov Chain Method. We set an initial position vector $\matr{u^{(1)}}$ and use it to calculate $\matr{u^{(n)}}$ from equation \ref{eq:l}, with which we can reversely update $\matr{u^{(1)}}$ by equation $\ref{eq:r}$. We repeat the above process until it converges, i.e. both $\matr{u^{(1)}}$ and $\matr{u^{(n)}}$ remain unchanged in iteration. This indicates that each vertex in $V$ is placed on both its left and right barycentres, meaning that $B_{L}(v_j) = B_{R}(v_j)$ for any $v_j \in V$. From equation \ref{eq:two_sided_barycentre}, we have $B(v_j) = B_{L}(v_j) = B_{R}(v_j)$ for all $v_j$ defined above, showing that we are able to achieve a barycentre ordering with this process.

Let $\matr{L} = \prod_{i = 1}^{n - 1}{\matr{\widetilde{L}^{(n - i)}}}$ of size $|V_1| \times |V_n|$ and $R = \prod_{i = 1}^{n - 1}{\matr{\widetilde{R}^{(i)}}}$ of size $|V_n| \times |V_1|$, we simplify the above process as
\begin{equation}
    \label{eq:not_markov}
    \matr{u^{(1)}} = \matr{R}\matr{u}_n = \matr{R}\matr{L}\matr{u^{(1)}} \coloneqq \matr{T} \matr{u^{(1)}}.
\end{equation}
Here we define \emph{transition matrix} $\matr{T} = \matr{R}\matr{L}$. We show in the following that it is a \emph{right stochastic matrix}, i.e. a non-negative real square matrix with each row summing to 1 that are used to represent the probabilities in the transition of a Markov chain. Firstly, $\matr{T}$ is of size $|V_1| \times |V_1|$ and therefore a square matrix. Moreover, by Property 2 of $L_i$ and $R_i$, all entries in $L_i$, $R_i$ and consequently $L$, $R$ are obviously non-negative real values. It remains to show that $\matr{T}$ is a row-normalized matrix. To this end, we first prove the following lemma:

\begin{lemma}
    \label{lemma:multi}
    Let $l \times m$ $\matr{A}$, $m \times n$ $\matr{B}$ be two row-normalized matrices, then their product $\matr{P} = \matr{A}\matr{B}$ is also a row-normalized matrix.
\end{lemma}
\begin{proof}
    We denote $\matr{A}$ and $\matr{B}$ as
    \begin{equation}
        \matr{A} = \begin{bmatrix}
            a_{11} & a_{12} & a_{13} & \dots  & a_{1m} \\
            a_{21} & a_{22} & a_{23} & \dots  & a_{2m} \\
            \vdots & \vdots & \vdots & \ddots & \vdots \\
            a_{l1} & a_{l2} & a_{l3} & \dots  & a_{lm}
        \end{bmatrix}
        and \:
        \matr{B} = \begin{bmatrix}
            b_{11} & b_{12} & b_{13} & \dots  & b_{1n} \\
            b_{21} & b_{22} & b_{23} & \dots  & b_{2n} \\
            \vdots & \vdots & \vdots & \ddots & \vdots \\
            b_{m1} & b_{m2} & b_{m3} & \dots  & b_{mn}
        \end{bmatrix}.
    \end{equation}
    Then we have
    \begin{equation}
        p_{ij} = a_{i1}b_{1j} + \dots + a_{im}b_{mj}
    \end{equation}
    and the sum of the $i$-th row of $\matr{P}$ is
    \begin{equation}
        \begin{array}{cc}
        \sum^{n}_{j=1} p_{ij} & = \sum^{n}_{j=1} (a_{i1}b_{1j} + \dots + a_{im}b_{mj})\\ & = a_{i1} \sum^{n}_{j=1} b_{1j} + \dots + a_{im}\sum^{n}_{j=1} b_{mj}
            \\ &= a_{i1} + \dots + a_{im} = 1
        \end{array}
    \end{equation}
\end{proof}

With lemma \ref{lemma:multi}, it follows that the product of more than two row-normalized matrices is still row-normalized. Thus, we have that $\matr{L}$, $\matr{R}$ and therefore $\matr{T}$ are row-normalized matrices. This completes the proof that $\matr{T}$ is a right stochastic matrix. As a result, we rewrite equation \ref{eq:not_markov} in the form of a Markov chain
\begin{equation}
    \matr{u^{(1)}}_{k + 1} = \matr{T}\matr{u^{(1)}}_k
\end{equation}
where convergence is guaranteed.

To solve a Markov chain $\matr{\pi_{n+1}} = \matr{P}\matr{\pi_{n}}$, we need to find the stationary distribution $\pi$ which satisfies $\pi = \matr{P}\pi$, i.e. $\pi$ is invariant by the transition matrix $\matr{P}$. Normally, the stationary distribution is first right eigenvector $\matr{x}_1$ corresponding to the largest eigenvalue $\lambda_1 = 1$ of $\matr{T}$. However, in our case, all entries in $\matr{x}_1$ are identical, which means that all vertices in $V_1$ should be placed at the same position. Consequently, vertices in other levels will also be placed at the same position as the barycentre remains unchanged. The resultant order is not a barycentre ordering by our previous definition.

On the other hand, the second largest sign-less eigenvalue of $\matr{T}$, $\lambda_2$, gives a heuristic solution for the Markov chain. One example of the usage of $\lambda_2$ comes from \cite{Liu2011}, where the author aims to use Markov chain to solve the NP-complete state clustering problem. In their case, the eigenvector associated with the largest eigenvalue yields a trivial solution having all vertices in one cluster. The second largest eigenvalue, however, gives a eigenvector that generates a satisfying approximation of the proper clustering. In our case, we find that $\matr{x}_2$ is also a competent alternative to $\matr{x}_1$. The sign of an eigenvalue is insignificant here as it just inverts the resultant ordering upside down without changing the crossing of edges.

Given $\matr{u^{(1)}}$ from solving the Markov chain, we obtain all other position vectors by $\matr{u^{(i+1)}} = \matr{L^{(i)}}\matr{u^{(i)}}$ and subsequently the ordering $\sigma$. This completes the Markov Chain Method.

The Markov Chain Method cannot yield a complete barycentre ordering for the following two reasons: (1) the random component in transition matrix $T$ from modified left/right transition matrix affects the final ordering; (2) the use of the second largest eigenvector is not optimal in nature. We propose a solution for the first problem. We solve the second problem in Stage 2 using the Partition Refinement Method.

Adding a random component in the left/right transition matrix avoids the problem of having multiple vertices with the same barycentre for which an ordering is impossible to determine. However, the calculated barycentres incorporate the randomness and then pass on to the subsequent levels. As a result, output of the Markov chain method is different each time. In this case, we repeat the Markov Chain Method for a predefined $N$ times and each time we calculate and record the output ordering $\sigma_{k}$ and its weighted crossing $K(G(V, E, n, \sigma_{k}))$. Then we choose the ordering with minimum weighted crossing as the \emph{best-in-$N$} ordering of Stage 1. This is facilitated by the fact that there are various efficient methods for calculating the second largest eigenvalue and the corresponding eigenvector. While by the nature of randomness, the best-in-$N$ ordering is still not a complete barycentre ordering, we have that the larger $N$ is, the smaller weighted crossing number the best-in-$N$ ordering has.

\subsection{Stage 2: Partition Refinement Method}

In Stage 1, we view each vertex as a point in its level. As a result, it gives the position value of the corresponding point when calculating the barycentres of all its connected vertices. Stage 2, on the other hand, allows a vertex to give each connected vertex an individual position value within a certain range for the calculation of barycentre. That is, instead of using a point to represent a vertex, each vertex is given a range, called \emph{block}, within the level while the end of each edge connecting this vertex is symbolized as a point in this given range. For $v_k \in N(v_j)$, we use $P(v_j, v_k)$ to denote the point of edge on the block of $v_k$ and $p(v_j, v_k)$ for the position value of $P(v_j, v_k)$.

Apparently, $p(v_j, v_k)$ should be dependent on the positions and therefore orders of $v_j$ and $v_k$ in their own levels. First of all, the $i$-th level is splitted into $|V_i|$ blocks of equal height and the $j$-th block from top to bottom is assigned to the $j$-th vertex in the ordering $\sigma_i$. Thus, the range of $p(v_j, v_k)$ is attribute to the order of $v_j$. Furthermore, the value of $p(v_j, v_k)$ within the block follows from the position of $v_k$, specifically the value of $P(v_k, v_j)$. For the $j$-th block in the $i$-th level, it has base $b_{i,j} = \frac{|V_i| - 1}{|V_i|}$ and height $h_i = \frac{1}{|V_i|}$. We obtain the initial position of $P(v_j, v_k)$ from the order of $v_j$ in best-in-$N$ ordering from Stage 1. For $v_q$, $v_p \in N_L(v_j)$, ordering of points $P(v_j, v_q)$ and $P(v_j, v_p)$ complies with that of $v_q$ and $v_p$. Ordered points of vertices in $N_L(v_j)$ are then distributed evenly within the block from top to bottom, dividing the block into $|N_L(v_j)|+1$ equal segments. Points of vertices in $N_R(v_j)$ follow the same. Let $v_j$ has order $\lambda$ in $\overline{\sigma_i}$ and a connected vertex $v_k$ has order $\mu$ in $N_L(v_j)$, then
\begin{equation}
    \label{eq:p_left}
    p(v_j, v_k) = b_{i, \lambda} + \frac{|N_L(v_j)|+1-\mu}{|N_L(v_j)|+1} h_i.
\end{equation}
If $v_k$ belongs to $N_R(v_j)$, Equation $\ref{eq:p_left}$ uses $|N_R(v_j)|$ instead of $N_L(v_j)$.

With $\overline{\sigma}$ from Stage 1, we can assign the vertices based on their order in the level. Subsequently, we obtain initial values for all points with the above equation. Then the iteration begins with calculating the barycentre for each of the middle levels from left to right. Starting with the second level, we obtain the barycentre of each vertex from both the first and the third levels. For the calculation of $B_L(v_j)$, rather than using the position vector $\matr{u^{(i-1)})}$ as in Equation $\ref{eq:bc_l}$, we use a position vector unique for this $v_j$ where
\begin{equation*}
    \matr{l}(v_j) = \left\{
\begin{array}{lcl}
p(v_k, v_j)      &      &  v_k \in N_L(v_j)\\
0       &      & v_k \in V_i \backslash N_(v_j).
\end{array} \right.
\end{equation*}
And we have
\begin{equation*}
B_L(v_j) = \matr{\widetilde{L}^{(i-1)}_{j,:}} \cdot \matr{l}(v_j),
\end{equation*}
\begin{equation*}
B_R(v_j) = \matr{\widetilde{R}^{(i)}_{j,:}} \cdot \matr{r}(v_j).
\end{equation*}
Then the two-sided barycentre can be derived from Equation \ref{eq:two_sided_barycentre}. With all barycentres of vertices in the second layer calculated, we have a new ordering $\sigma_2$ complying with the descending ordering of the barycentres. On the basis of the new $\sigma_2$, we reassign a block for each vertex. Asides from the change in the range of value, new position of a point $P(v_j, v_k)$ is also under the affect of point $P(v_k, v_j)$'s position. Instead of evenly distributing the points within the block as before, we let $P(v_j, v_k)$ takes the value of $p(v_k, v_j)$ after being scaled to the height of the block. Thereupon, let the order of $v_j$ in the new $\sigma_2$ be $\lambda^\prime$, Equation $\ref{eq:p_left}$ becomes
\begin{equation*}
    p(v_j, v_k) = b_{i, \lambda^\prime} + h_i p(v_k, v_j).
\end{equation*}

The calculation of the remaining middle levels follows. For each vertex in the last level, its barycentre is equal to its left barycentre. After updating the ordering and points for the last level, we go backwards and repeat the operation on the middle levels from right to left. Procedure on the first level is similar to that of the last. This completes one iteration.

This iteration converges to a complete barycentre ordering when the ordering remains unchanged for a predetermined $M$ times. However, as mentioned before, a complete barycentre ordering is not necessarily an optimal one. It is therefore possible that we encounter an ordering with less weighted crossing than the final complete barycentre ordering. In light of this possibility and our ultimate aim of reducing weighted crossing, we keep a record of the weighted crossing from each iteration if the resultant ordering is different from its predecessor. This allows us to choose the ordering with minimum weighted crossing produced in the process as the output ordering of Stage 2.

\section{Modified Method on the cycle form}

In this section, we modify the method introduced in previous section to reduce weighted crossing of the specified cycle form of Sankey diagram. We start by formulating the cycle form of the Sankey diagram as a \emph{circular layer graph}. We refer the links connecting the last and the first level as the \emph{binding links}. By ignoring the binding links, the cycle form becomes a parallel one with $\overline{G} = (V, \overline{E} = \{E_i | i \in [1, n-1]\}, n)$. On the other hand, for the binding links, we use $E_n$ to denote corresponding the edge set.  Subsequently, we have the graph for the original cycle form by adding back $E_n$, i.e. $G = (V, E = \overline{E} \bigcup \{E_n\}, n)$.

For Stage 1, we supply $\overline{G}$ to the Markov Chain method. In this way, Stage 1 yields the best-in-$N$ ordering without considering the binding links. Subsequently, we utilize Stage 2 to take the binding links back into consideration. 

In Stage 2, the first modification we make to the Partition Refinement method is on the calculation of barycentre. In the previous section, the barycentre of vertex in $V_1$ is just the barycentre of its right neighboring set. Now each vertex in $V_1$ also has a left neighboring set consisting of connected vertices in the last level, allowing the Equation \ref{eq:two_sided_barycentre} to be applicable. Similarly, vertices in the last level also have their own right neighboring set for calculating the barycentres.

Given the circular nature of graph $G$, the route of iteration also changes. Instead of adopting the "back and forth" manner in the previous section, after obtaining a new $\sigma_n$, we proceed to apply the method on the first level. The modified Partition Refinement method is summarized as pseudo code in Algorithm \ref{algo:modified}.

\begin{algorithm}[htb]
\caption{ Modified Partition Refinement Method }
\label{algo:modified}
\begin{algorithmic}[1]
\REQUIRE ~~\\
The best-in-$N$ ordering $\overline{\sigma}$ from Stage 1.
The maximum repeat number $M$.
\ENSURE ~~\\
Improved ordering $\Tilde{\sigma}$.
\FOR{each $m \in [1, M]$}
\FOR{each $i \in [1, n]$}
    \STATE{Use $\sigma_{i-1}$ and $\sigma_{i+1}$ to update $\tau^{(i)}$ and $\epsilon^{(i)}$;}
    \label{code:pr:m_sub_order}
    \STATE{Use $\tau^{(i)}$ and $\epsilon^{(i)}$ to update $\matr{r^{(i)}}$ and $\matr{l^{(i)}}$;}
    \label{code:pr:m_sub_pos}
    \STATE{Calculate the barycentres of vertices in $V_i$ and update the position vertex $\matr{u^{(i)}}$;}
    \label{code:pr:m_pos}
    \STATE{Use new $\matr{u^{(i)}}$ to update ordering
    $\sigma_{i}$;}
    \label{code:pr:m_order}
\ENDFOR
\ENDFOR
\RETURN refined ordering 
\end{algorithmic}
\end{algorithm}

\section{Result}

We prove the efficiency of our method on the \emph{parallel form} by comparing our result with that of the state-of-art heuristic method in \cite{Alemasoom2016}. Moreover, we show that our method produces near-optimal result by demonstrating a low contrast between the our performance and that of the optimal ILP method in \cite{optimal}. Besides from visual comparison, we also measure the performances of the above methods by both the weighted and non-weighted crossing of the output ordering. In particular, resultant orderings from the three comparing methods are obtained from the corresponding Sankey diagrams provided in their articles. To increase the sensitivity of our method towards edges with considerably small weights, we use logarithm (with base 10) for the edge weights in the above tests to reduce their difference. For our adapted method on the \emph{cycle form}, we apply it on a graph with zero crossing to see if it can achieve the optimal result in this case. We also conduct a robust test to validate the consistency of our method against varying complexity of graph.

\subsection{Test against State-of-Art Heuristic Method}

In \cite{Alemasoom2016}, they applied their method on Canada's energy usage data in 1978 and Figure \ref{fig:heurResult} displays their resultant Sankey diagram. We apply our method on the same dataset. In Stage 1, we set $\alpha_1 = 0.1$ and $N = 100$. In Stage 2, we set $\alpha_2 = 0.1$, $M = 100$ and obtain the convergent result with only 10 iterations. Figure \ref{fig:heur1} and Figure \ref{fig:heur2} give the Sankey diagrams with orderings from Stage 1 and 2 of our algorithm respectively. Table \ref{tab:heur} summarizes the weighted and non-weighted crossing of orderings from all three Sankey diagrams in Figure \ref{fig:heur}.

\begin{figure}
     \centering
     \begin{subfigure}[b]{\textwidth}
         \centering
         \includegraphics[height=4cm]{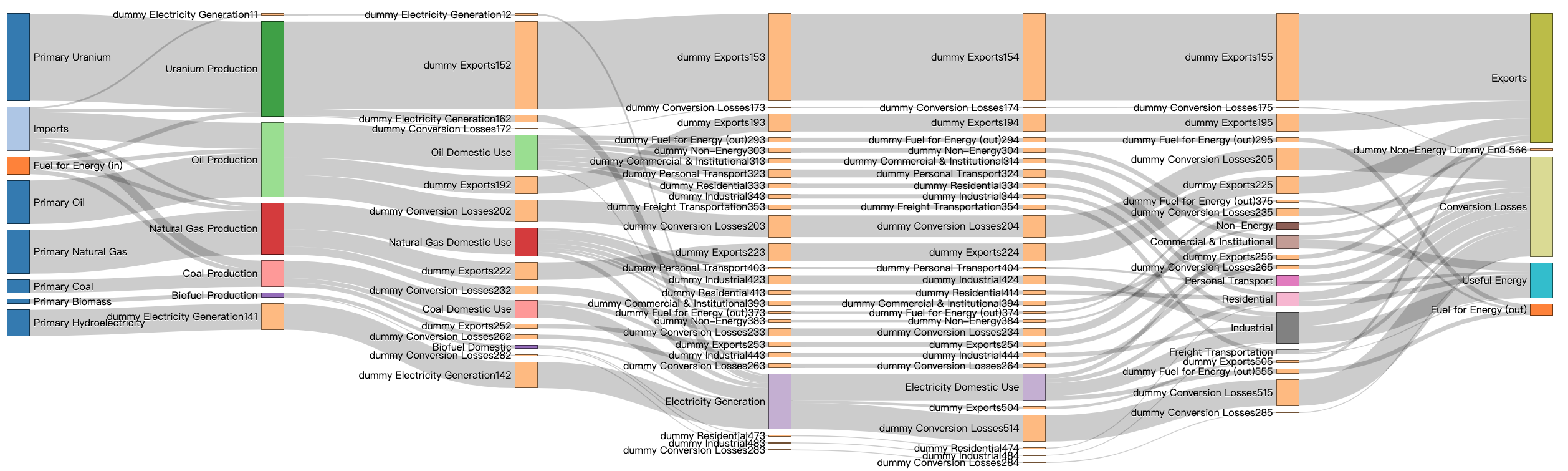}
         \caption{Sankey diagram produced by \cite{Alemasoom2016}}
         \label{fig:heurResult}
     \end{subfigure}
     \hfill
     \begin{subfigure}[b]{\textwidth}
         \centering
         \includegraphics[height=4cm]{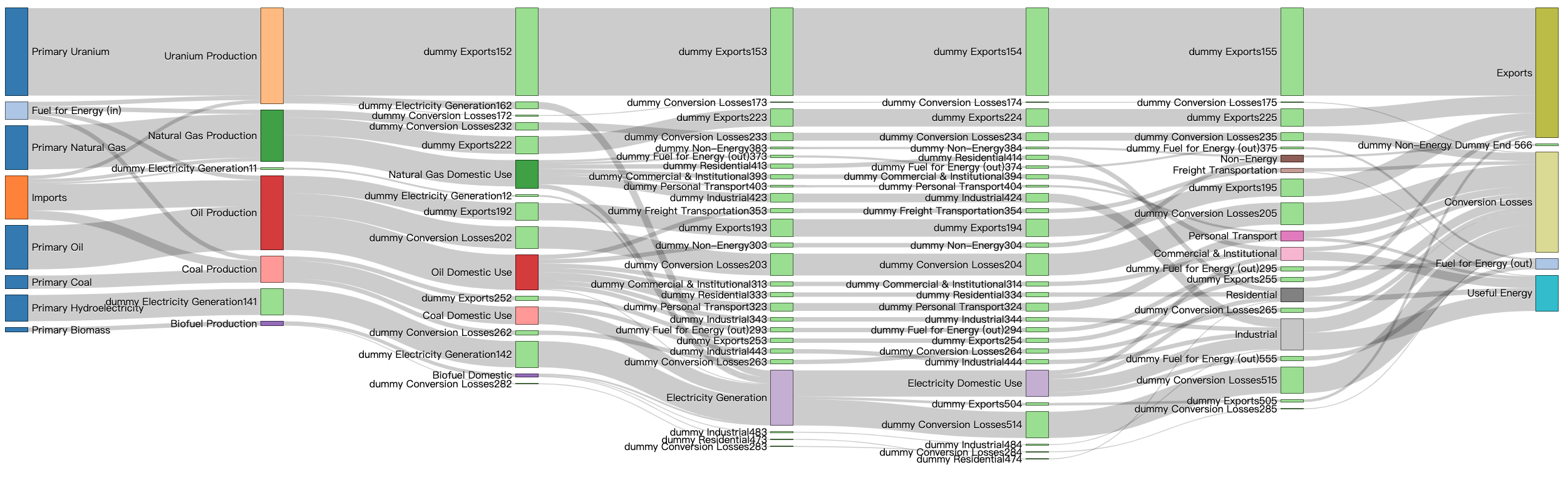}
         \caption{Sankey diagram produced by Stage 1}
         \label{fig:heur1}
     \end{subfigure}
     \hfill
     \begin{subfigure}[b]{\textwidth}
         \centering
         \includegraphics[height=4cm]{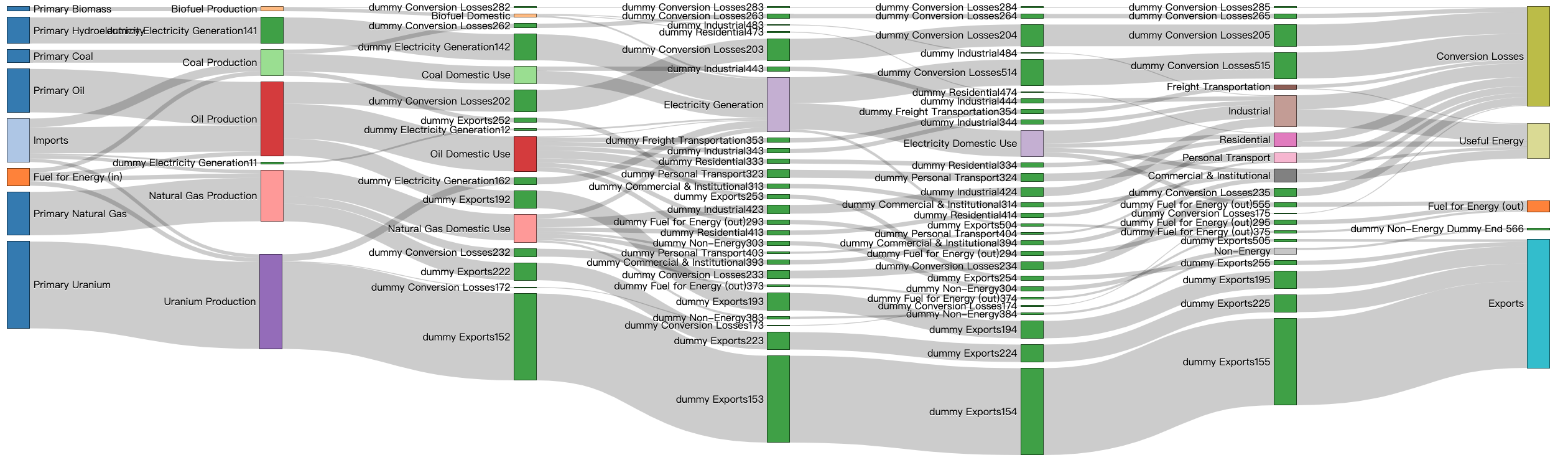}
         \caption{Sankey diagram produced by Stage 2}
         \label{fig:heur2}
     \end{subfigure}
        \caption{Three Sankey diagrams corresponding the state-of-art heuristic method and two stages of our method}
        \label{fig:heur}
\end{figure}

\begin{table}[h]
    \centering
    \begin{tabular}{l | l | l | l}
    Ordering from:& Figure \ref{fig:heurResult} & Figure \ref{fig:heur1} & Figure \ref{fig:heur2} \\
    \hline
    Weighted Crossing & 146.77 & 112.59 & 87.855 \\
    Non-Weighted Crossing & 279 & 226 & 158.0
    \end{tabular}
    \caption{Summarization of measurements of orderings from the three Sankey diagrams in \ref{fig:heur}}
    \label{tab:heur}
\end{table}

Base on the comparison of both crossing measurements, we see that even without the refinement in Stage 2, the output from Stage 1 already surpass that of the heuristic method. And the improvement in both measurements from Stage 1 to Stage 2 validates the effectiveness of Stage 2. To be specific, we find that Stage 2 is able to resolve some unnecessary crossings in Stage 1 resulting from the additional random component. On the other hand, Stage 1 gives a satisfying semi-barycentre ordering such that Stage 2 only takes a few iterations to converge.

\subsection{Test against ILP Method and BC Method}

This test examines the difference between output from our method and the optimal output from the ILP method in \cite{optimal} to see if we have a near-optimal result. Moreover, we compare our result against that of the BC method which shares our idea of finding a barycentre ordering such that we can can demonstrate our ability to produce a better barycentre ordering.

In this test, we use the same dataset as in \cite{optimal}: the "World Greenhouse Gas Emissions" data from the World Resource Institute \cite{greenhousedata}. From \cite{optimal} we have results of both ILP method and the BC method on this dataset. For our method, we supply $\alpha_1 = 0.01$, $N=100$ to Stage 1 and $\alpha_2 = 0.1$, $M = 100$ to Stage 2 which converges with only one iteration. For the output orderings of the two comparing methods and the two stages of our method, we calculate their weighted and non-weighted crossings and summarize them in Table \ref{tab:ilp}. We also plot the four orderings as Sankey diagrams in \ref{fig:ilp} for visual comparison.

Comparing Figure \ref{fig:ILP_optimal} and Figure \ref{fig:ILP_stage_2}, we see that we are still at a relatively small distance from the optimal layout. Table \ref{tab:ilp} shows that although the output from Stage 2 has less non-weighted crossing number but still larger weighted crossing number.

On the other hand, our output excels from that of the BC method from both visual aspect and two crossing measurements. Also, BC method is an iterative method and therefore requires time to achieve a near-optimal ordering. In contrast, the output from Stage 1 already suffices as a near-optimal ordering. Moreover, iteration times $M$ to convergence in Stage 2 is also small.

\begin{table}[h]
    \centering
    \begin{tabular}{l | l | l | l | l}
    Ordering from:& Figure \ref{fig:ILP_heuristic} & Figure \ref{fig:ILP_optimal} & Figure \ref{fig:ILP_stage_1} & Figure \ref{fig:ILP_stage_2}\\
    \hline
    Weighted Crossing & 1220.07 & 156.64 & 300.89 & 278.68 \\
    Non-Weighted Crossing & 322 & 125 & 126 & 121
    \end{tabular}
    \caption{Table summarizing the measurements of orderings from the three Sankey diagrams in Figure \ref{fig:ilp}.}
    \label{tab:ilp}
\end{table}

\begin{figure}
     \centering
     \begin{subfigure}[b]{\textwidth}
         \centering
         \includegraphics[height=4cm]{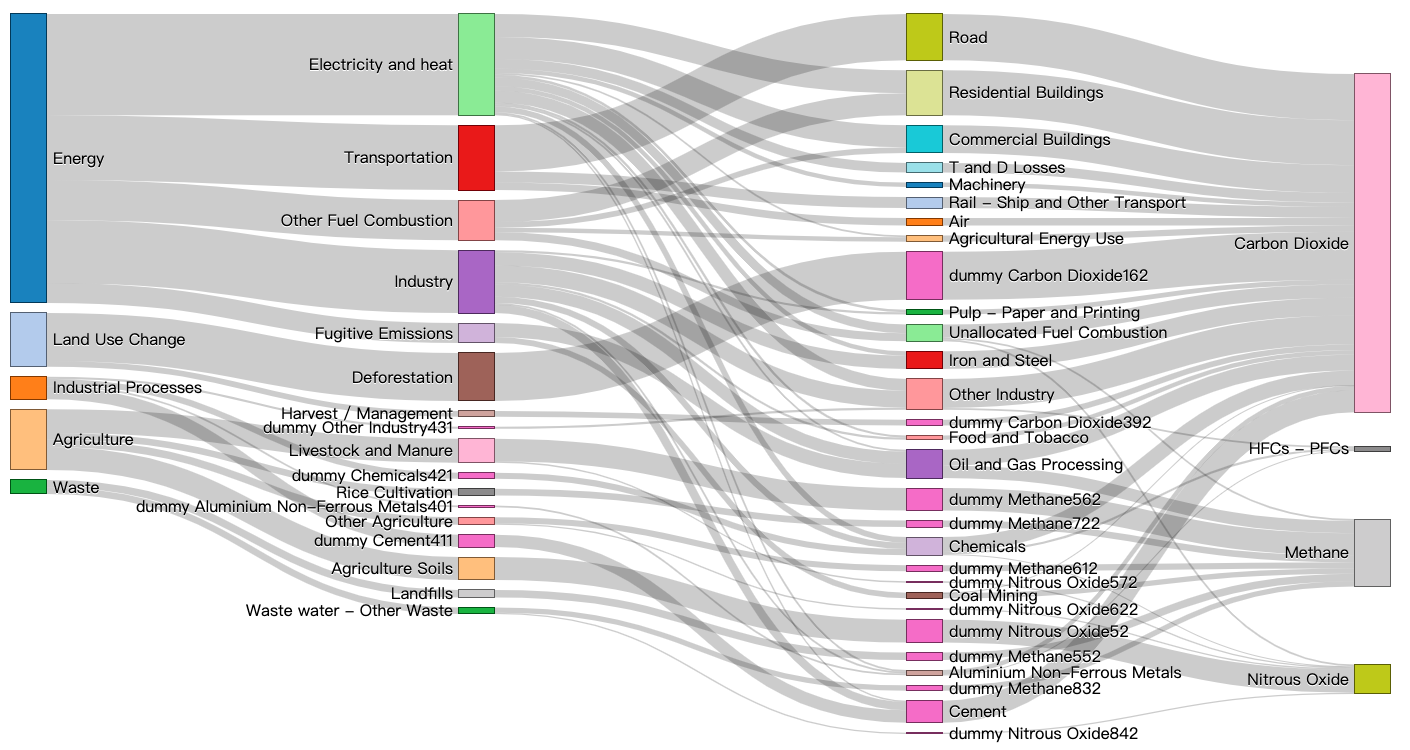}
         \caption{Sankey diagram produced by the BC method}
         \label{fig:ILP_heuristic}
     \end{subfigure}
     \hfill
     \begin{subfigure}[b]{\textwidth}
         \centering
         \includegraphics[height=4cm]{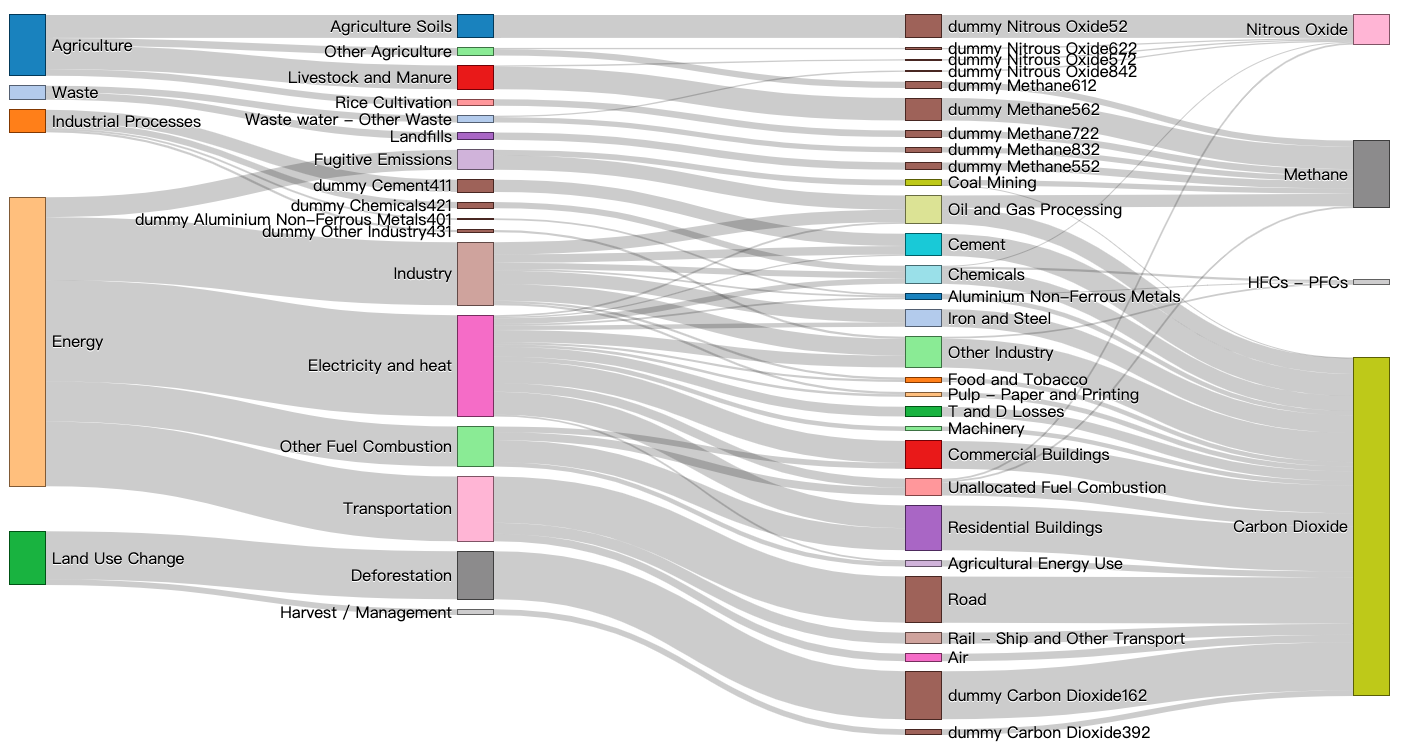}
         \caption{Sankey diagram produced by the ILP method}
         \label{fig:ILP_optimal}
     \end{subfigure}
     \hfill
     \begin{subfigure}[b]{\textwidth}
         \centering
         \includegraphics[height=4cm]{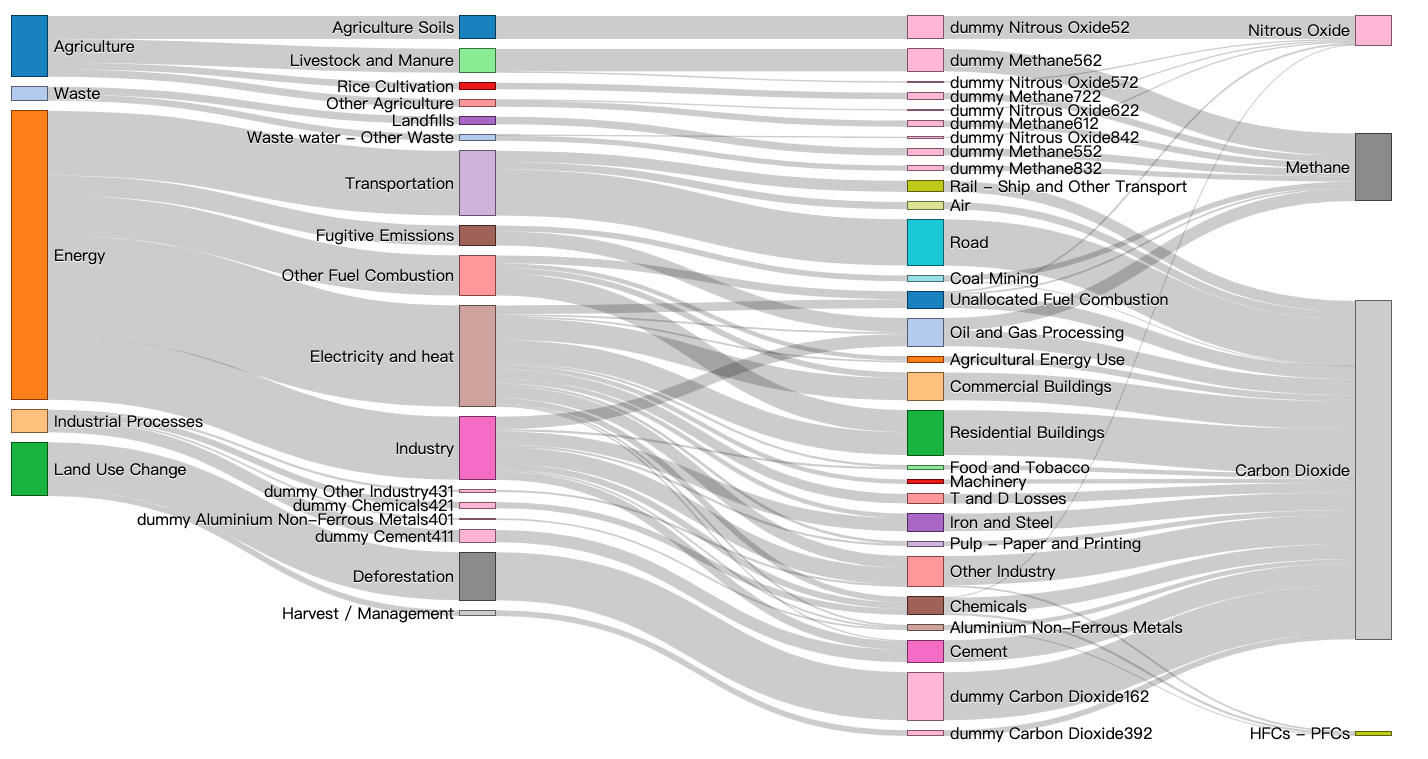}
         \caption{Sankey diagram produced by Stage 1}
         \label{fig:ILP_stage_1}
     \end{subfigure}
     \hfill
     \begin{subfigure}[b]{\textwidth}
         \centering
         \includegraphics[height=4cm]{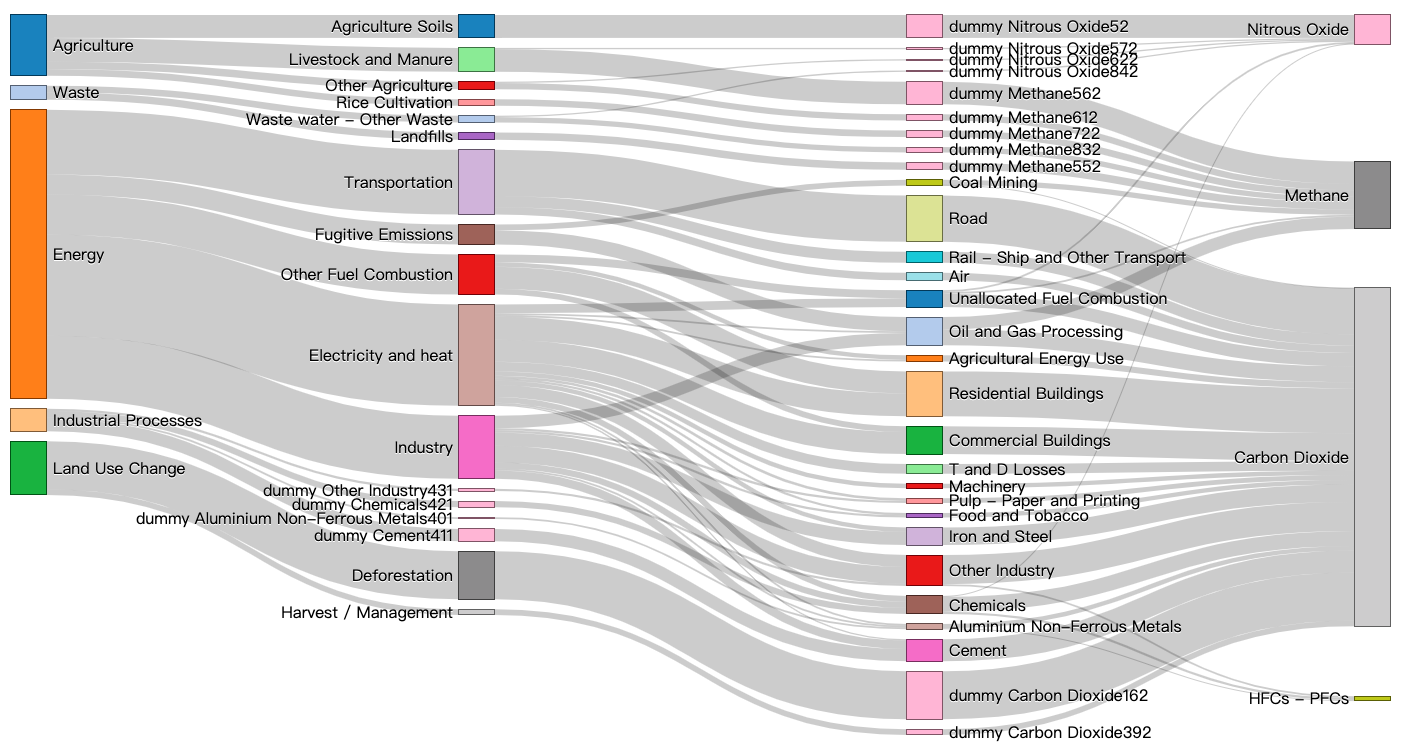}
         \caption{Sankey diagram produced by Stage 2}
         \label{fig:ILP_stage_2}
     \end{subfigure}
    \caption{Four graph outputs}
    \label{fig:ilp}
\end{figure}

\subsection{Test of Modified Method on the Cycle Form}

In this test we apply the modified algorithm on an artificial dataset of \emph{cycle form}. This dataset has the same number of level and same number of vertices in each level as the dataset in the previous test against the ILP method and the BC method. In particular, this artificial data has a known optimal layout which has zero crossing number. 

In this test case, we find that Stage 1 alone is able to produce an optimal layout with $N=100$, $\alpha_1 = 0.01$. To test the effectiveness of Stage 2, we reset $N = 50$ in Stage 1 and obtain an ordering of $K = 4$ and $\overline{K} = 2$. Stage 2 then refines the Best-in-N ordering to the optimal ordering within 2 iterations with $\alpha_2 = 0.01$.

\subsection{Robust Test}

With this test, we aim to show the stability of our method towards cases of different complexity level. The complexity level of a graph is measured by the number of level (denoted as $n$) and the number of vertex in each level (denoted as $\overline{V}$). Consequently, we vary both $n$ and $\overline{V}$ and for each pair we generate ten different random cases. The total edge number of a random case is considered an estimation of its complexity. All test case run with $\alpha_1 = 0.01$, $N=100$ for Stage 1 and $\alpha_2 = 0.1$, $M = 100$ for Stage 2. We record for each case the weighted crossing produced by both Stage 1 and Stage 2, as well as the optimal result from the ILP method for comparison. Let a result from our method be $x$ and the corresponding optimal result be $y$, we measure their difference by a ratio 
\begin{equation}
    \textit{r} = \frac{x + \epsilon}{y + \epsilon}
\end{equation}
where $\epsilon$ is a very small number to avoid the denominator being 0 when $y=0$. The reason we do not use the difference between $x$ and $y$ directly is because as the complexity of the graph increase, the resultant weighted crossing and therefore the difference are also increasing.

\begin{figure}[h]
    \centering
    \includegraphics[width=\textwidth]{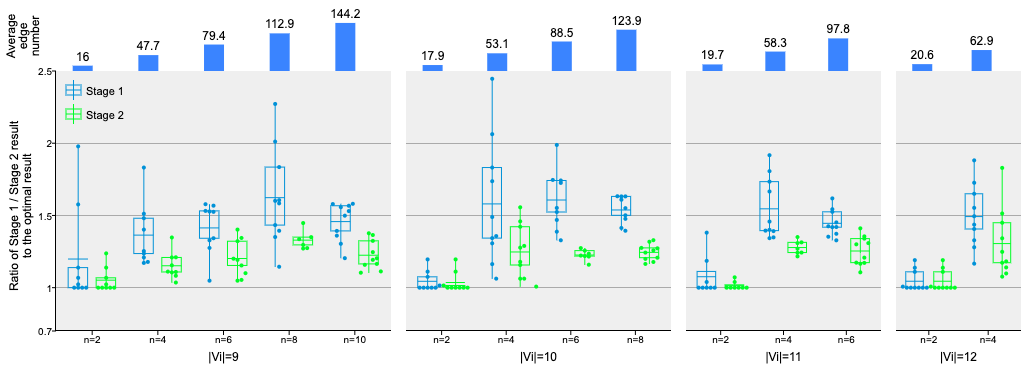}
    \caption{Summarization of result for the robust test}
    \label{fig:robust}
\end{figure}

We summarize the result of the robust test in Figure $\ref{fig:robust}$. For starters, it shows that most of result from Stage 1 are less than twice of the ILP result. Stage 2 further improves the result to be no more than 1.5 times of the ILP result. We believe this demonstrates the consistency of our method's performance on cases with various complexity.

\section{Conclusion}

In this paper, we investigate the NP-hard weighted crossing reduction problem for the Sankey diagram. Other than the common parallel form of Sankey diagram, we also study a particular circular form where the first and the last layers are connected.


Our heuristic method, aiming to find a barycentre ordering, are composed of two stages. The first stage employs the Markov Chain method and the second stage serves to improve Stage 1's output with the Partition Refinement Method. We also adapt this method to be applicable for reducing weighted crossing in the specified circular form of the diagram.

After experiments, we can conclude that our method performs nearly as well as the ILP method and surpasses the existing heuristic methods. In terms of the measurement of weighted crossing, in the ILP experiment, our method achieved 300.89 weighted crossings, very close to the 278.68 weighted crossings from the ILP method while the BC method has a weighted crossing number of 1220.07. Also, we obtained only 87.855 weighted crossings while the stage-of-art heuristic method attained 146.77. Visually speaking, we were able to obtain high readability even from complicated seven-layer data. We also performed a robust test which verified the stability of our method against changing complexity of dataset.

\bibliographystyle{plain}
\bibliography{ref}
\end{document}